\documentclass[a4paper]{article}

\usepackage{latexsym}
\usepackage{amssymb,amsmath,,amsthm}
\usepackage{url}
\usepackage{bbm}

\usepackage{tikz}

\usetikzlibrary{positioning,arrows,fit}
\usetikzlibrary{decorations.pathmorphing}

\usepackage{color}

\newtheorem{definition}{Definition}
\newtheorem{theorem}[definition]{Theorem}
\newtheorem{lemma}[definition]{Lemma}

\newtheorem{question}[definition]{Question}
\newtheorem{remark}[definition]{Remark}

\newcommand{\OWL}{\textsl{OWL\,2}}
\newcommand{\OWLQL}{\textsl{OWL\,2\,QL}}
\newcommand{\OWLEL}{\textsl{OWL\,2\,EL}}

\newcommand{\DL}{\textsl{DL-Lite}}
\newcommand{\q}{\boldsymbol q}

\newcommand{\NP}{\textsc{NP}}

\newcommand{\Ppoly}{\textsc{P}/\text{poly}}

\newcommand{\PSpace}{\textsc{PSpace}}
\newcommand{\PH}{\textsc{PH}}

\newcommand{\ind}{\mathop{\mathsf{ind}}}

\newcommand{\cli}{\textsc{Clique}}
\newcommand{\mat}{\textsc{Matching}}
\newcommand{\gen}{\textsc{Gen}}

\newcommand{\A}{\mathcal{A}}
\newcommand{\T}{\mathcal{T}}

\newcommand{\I}{\mathcal{I}}

\newcommand{\CIK}{\mathcal{C}_\mathcal{K}}

\newcommand{\C}{\mathbf{C}}

\newcommand{\phifn}[1][\vec{\alpha}]{\varphi_{f^n}^{#1}}

\begin{document}

\title{Exponential Lower Bounds and Separation for Query Rewriting}
\author{S.~Kikot,\!$^1$  R.~Kontchakov,\!$^1$ V.~Podolskii$^2$ and M.~Zakharyaschev$^1$\\[5pt]
    $^1$ Department of Computer Science and Information Systems\\
   Birkbeck, University of London, U.K.\\
   \url{{kikot,roman,michael}@dcs.bbk.ac.uk}\\
   $^2$ Steklov Mathematical Institute, Moscow, Russia\\
   \url{podolskii@mi.ras.ru}}


\maketitle

\begin{abstract}
We establish connections between the size of circuits and formulas computing monotone Boolean functions and the size of first-order and nonrecursive Datalog rewritings for conjunctive queries over \OWLQL{} ontologies. We use known lower bounds and separation results from  circuit complexity to prove similar results for the size of rewritings that do not use non-signature constants. For example, we show that, in the worst case, positive existential and nonrecursive Datalog rewritings are exponentially longer than the original queries; nonrecursive Datalog rewritings are in general exponentially more succinct than positive existential rewritings; while first-order rewritings can be superpolynomially more succinct than positive existential rewritings.
\end{abstract}


\section{Introduction}

First-order (FO) rewritability is the key concept of ontology-based data access (OBDA)~\cite{DFKM*08,HMAM*08,PLCD*08}, which is believed to lie at the foundations of the next generation of information systems. An ontology language $\mathcal{L}$ enjoys \emph{FO-rewritability} if any conjunctive query $\q$ over an ontology $\mathcal{T}$, formulated in $\mathcal{L}$, can be transformed into an FO-formula $\q'$ such that, for any data $\mathcal{A}$, all certain answers to $\q$  over the knowledge base $(\mathcal{T}, \mathcal{A})$ can be found by querying $\q'$ over $\mathcal{A}$ using a standard relational database management system (RDBMS). Ontology languages with this property include the \OWLQL{} profile of the Web Ontology Language \OWL, which is based on description logics of the \DL{} family~\cite{CDLLR07,ACKZ09}, and fragments of Datalog$^\pm$ such as linear or sticky sets of TGDs~\cite{CaliGL09,CaliGP10}. Various rewriting techniques have been implemented  in the systems QuOnto~\cite{ACDL*05}, REQUIEM~\cite{Perez-UrbinaMH09}, Presto~\cite{RosatiAKR10}, Nyaya~\cite{2011_Gottlob}, IQAROS\footnote{\url{http://code.google.com/p/iqaros/}} and Quest\footnote{\url{http://obda.inf.unibz.it/protege-plugin/quest/quest.html}}.

OBDA via FO-rewritability relies on the empirical fact that RDBMSs are usually very efficient in practice. However, this does not mean that they can efficiently evaluate any given query: after all, for expression complexity, database query answering is \PSpace-complete for FO-queries and \NP-complete for conjunctive queries (CQs).
Indeed, the first `na\"ive' rewritings of CQs over \OWLQL{} ontologies turned out to be too lengthy even for modern RDBMSs~\cite{CDLLR07,Perez-UrbinaMH09}. The obvious next step was to develop various optimisation techniques~\cite{RosatiAKR10,2011_Gottlob,DL-2011-obda,ISWC-2011}; however, they still produced exponential-size --- $O((|\mathcal{T}| \cdot |{\q}|)^{|\q|})$ --- rewritings in the worst case. An alternative two-step \emph{combined approach} to OBDA with \OWLEL{}~\cite{EL} and \OWLQL{}~\cite{KR10our,IJCAIbest} first expands the data by applying the ontology axioms and introducing new individuals required by the ontology, and only then rewrites the query over the expanded data. Yet, even with these extra resources a simple polynomial rewriting was constructed only for the fragment of \OWLQL{} without role inclusions; the rewriting for the full language remained exponential.
A breakthrough seemed to come in~\cite{GottlobS11}, which showed that one can construct, in polynomial time, a nonrecursive Datalog rewriting for some fragments of Datalog$^\pm$ containing \OWLQL. However, this rewriting uses the built-in predicate $\ne$ and numerical constants that are not present in the original query and ontology. Without such  additional constants, as shown in~\cite{KKZ-DL11}, no FO-rewriting for \OWLQL{} can be \emph{constructed} in polynomial time (it remained unclear, however, whether such an FO-rewriting of  polynomial size \emph{exists}). 

These developments bring forward a spectrum of theoretical and practical questions that could influence the future of OBDA. What is the worst-case size of FO- and nonrecursive Datalog rewritings for CQs over \OWLQL{} ontologies? What is the type/shape/size of rewritings we should aim at to make OBDA with \OWLQL{} efficient? What extra means (e.g., built-in predicates and constants) can be used in the rewritings?

In this paper, we investigate the worst-case size of FO- and nonrecursive Datalog rewritings for CQs over \OWLQL{} ontologies depending on the available means. We distinguish between `pure' rewritings, which cannot use constants that do not occur in the original query, and `impure' ones, where such constants are allowed.  Our results can be summarised as follows:
\begin{itemize}\itemsep=0pt
\item[--] An exponential blow-up is unavoidable for pure positive existential rewritings and pure nonrecursive Datalog rewritings. Even pure FO-rewritings with $=$ can blow-up superpolynomially unless $\NP \subseteq \Ppoly$.

\item[--] Pure nonrecursive Datalog rewritings are in general exponentially more succinct than pure positive existential rewritings.

\item[--] Pure FO-rewritings can be superpolynomially more succinct than pure positive existential rewritings.

\item[--] Impure positive existential rewritings can always be made polynomial, and so they are exponentially more succinct than pure rewritings.
\end{itemize}
We obtain these results by first establishing connections between pure rewritings for conjunctive queries over \OWLQL{} ontologies and circuits for monotone Boolean functions, and then using known lower bounds and separation results for the circuit complexity of such functions as $\cli_{n,k}$ `a graph with $n$ nodes contains a $k$-clique' or $\mat_{2n}$ `a bipartite graph with $n$ vertices in each part has a perfect matching.'
%


\section{Queries over \OWLQL{} Ontologies}

By a \emph{signature}, $\Sigma$, we understand in this paper any set of constant symbols and predicate symbols (with their arity). Unless explicitly stated otherwise, $\Sigma$ does \emph{not} contain any predicates with fixed semantics, such as $=$ or $\ne$. In the description logic (or \OWLQL{}) setting, constant symbols are called \emph{individual names}, $a_i$, while unary and binary  predicate symbols are called  \emph{concept names}, $A_i$, and \emph{role names}, $P_i$, respectively, where $i\ge 1$.

The language of \OWLQL{} is built using those names in the following way. The \emph{roles} $R$, \emph{basic concepts} $B$ and \emph{concepts} $C$  of \OWLQL{} are defined by the grammar:\footnote{We do not consider data properties, attributes and role (ir)reflexivity constraints.} 
\begin{align*}
R \quad &::=\quad P_i \quad\mid\quad P_i^-, \tag{\mbox{\small $P_i(x,y) \ \mid \  P_i(y,x)$}} \\
B \quad &::=\quad \bot \quad\mid\quad
A_i \quad\mid\quad \exists R, \tag{\mbox{\small $\bot \  \mid \  A_i(x) \ \mid \ \exists y\, R(x,y)$}}\\
C \quad &::=\quad B \quad\mid\quad \exists R.B, \tag{\mbox{\small $B(x) \ \mid \ \exists y\, (R(x,y) \land B(y))$}}
\end{align*}
where the formulas on the right give a first-order translation of the \OWLQL{} constructs.
An \OWLQL{} \emph{TBox}, $\mathcal{T}$, is a finite set of \emph{inclusions} of the form
\begin{align*}
&B \sqsubseteq C, \tag{\mbox{\small $\forall x \, (B(x) \to C(x))$}}  \\
&R_1 \sqsubseteq R_2, \tag{\mbox{\small $\forall x,y \, (R_1(x,y) \to R_2(x,y))$}}\\
& B_1 \sqcap B_2 \sqsubseteq \bot, \tag{\mbox{\small $\forall x \, (B_1(x) \land B_2(x) \to \bot)$}}\\
& R_1 \sqcap R_2 \sqsubseteq \bot. \tag{\mbox{\small $\forall x,y \, (R_1(x,y) \land R_2(x,y) \to \bot)$}}
\end{align*}
Note that concepts of the form $\exists R.B$ can only occur in the right-hand side of concept inclusions in \OWLQL.
An \emph{ABox}, $\mathcal{A}$, is a finite set of \emph{assertions} of
the form $A_k(a_i)$ and $P_k(a_i,a_j)$. $\mathcal{T}$ and $\mathcal{A}$ together form the \emph{knowledge base} (KB) $\mathcal{K}=(\mathcal{T},\mathcal{A})$. The semantics for \OWLQL{} is defined in the usual way \cite{BCMNP03}, based on interpretations $\mathcal{I} = (\Delta^\mathcal{I}, \cdot^\mathcal{I})$ with domain $\Delta^\mathcal{I}$ and interpretation function $\cdot^\mathcal{I}$.

The set of individual names in an ABox $\mathcal{A}$ will be denoted by $\ind(\mathcal{A})$.
For concepts or roles $E_1$ and $E_2$, we write $E_1 \sqsubseteq_\mathcal{T} E_2$  if $\mathcal{T} \models E_1 \sqsubseteq E_2$; and we set $[E] = \{ E' \mid E \sqsubseteq_\mathcal{T} E' \text{ and } E' \sqsubseteq_\mathcal{T} E \}$.

A \emph{conjunctive query} (CQ) $\q(\vec{x})$ is a first-order formula $\exists \vec{y}\, \varphi(\vec{x}, \vec{y})$,
where $\varphi$ is constructed, using $\land$, from atoms of the form $A_k(t_1)$ and $P_k(t_1,t_2)$, where each $t_i$ is a \emph{term} (an individual or a variable from  $\vec{x}$ or $\vec{y}$). A tuple $\vec{a}\subseteq \ind (\mathcal{A})$ is a \emph{certain answer} to $\q(\vec{x})$ over $\mathcal{K} = (\mathcal{T},\mathcal{A})$ if  $\mathcal{I} \models \q(\vec{a})$ for all models
$\mathcal{I}$ of $\mathcal{K}$; in this case we write $\mathcal{K} \models \q(\vec{a})$.

Query answering over \OWLQL{} KBs is based on the fact that, for any consistent KB $\mathcal{K} = (\mathcal{T}, \mathcal{A})$, there is an interpretation $\CIK$ such that, for all CQs $\q(\vec{x})$ and $\vec{a} \subseteq \ind(\mathcal{A})$, we have $\mathcal{K} \models \q(\vec{a})$ iff $\CIK \models \q(\vec{a})$.
The interpretation $\CIK$, called the \emph{canonical model} of $\mathcal{K}$, can be constructed as follows. For each pair $[R],[B]$ with $\exists R.B$ in $\mathcal{T}$ (we assume $\exists R$ is just a shorthand for $\exists R.\top$), we introduce a fresh symbol $w_{[RB]}$ and call it the \emph{witness for} $\exists R.B$. We write $\mathcal{K} \models C(w_{[RB]})$ if $\exists R^-  \sqsubseteq_\mathcal{T} C$ or $B \sqsubseteq_\mathcal{T} C$.
Define a \emph{generating relation}, $\leadsto$, on the set of these witnesses together with $\ind(\mathcal{A})$ by taking:
\begin{description}
\item[--] $a \leadsto w_{[RB]}$ if $a \in \ind(\mathcal{A})$, $[R]$ and $[B]$ are $\sqsubseteq_\mathcal{T}$-minimal such that $\mathcal{K} \models \exists R .B(a)$ and there is no $b\in\ind(\mathcal{A})$ with $\mathcal{K}\models R(a,b) \land B(b)$;

\item[--] $w_{[R'B']} \leadsto w_{[RB]}$ if, for some $u$, $u \leadsto w_{[R'B']}$, $[R]$ and $[B]$ are $\sqsubseteq_\mathcal{T}$-minimal with $\mathcal{K}\models \exists R.B(w_{[R'B']})$
    and it is not the case that $R' \sqsubseteq_\mathcal{T} R^-$ and $\mathcal{K} \models B'(u)$.
\end{description}
If $a\leadsto w_{[R_1B_1]} \leadsto \dots
\leadsto w_{[R_{n}B_{n}]}$, $n\ge 0$, then we say that $a$
\emph{generates the path} $a w_{[R_1B_1]} \cdots w_{[R_nB_n]}$.
Denote by $\mathsf{path}_\mathcal{K}(a)$ the set of paths generated by $a$, and by
$\mathsf{tail}(\pi)$ the last element in~$\pi \in \mathsf{path}_\mathcal{K}(a)$.
$\CIK$ is defined by taking:
\begin{align*}
\Delta^{\CIK} = & \bigcup_{a \in
\ind(\mathcal{A})}\mathsf{path}_\mathcal{K}(a), \quad
 a^{\CIK} ~=~  a, \text{ for } a \in \ind(\mathcal{A}), \\
 A^{\CIK} =~ &\{ \pi \in \Delta^{\CIK}\mid \mathcal{K} \models A(\mathsf{tail}(\pi)) \}, \\
 P^{\CIK} =~ & \{ (a,b)\in \ind(\mathcal{A})\times \ind(\mathcal{A}) \mid  \mathcal{K}\models P(a,b) \} \; \cup \\
 & \{ (\pi,\pi \cdot w_{[RB]}) \mid \mathsf{tail}(\pi) \leadsto w_{[RB]},\ R\sqsubseteq_\mathcal{T} P \} \cup{} \\
& \{(\pi\cdot w_{[RB]},\pi) \mid \mathsf{tail}(\pi) \leadsto w_{[RB]},\  R \sqsubseteq_\mathcal{T} P^-\}.%
\end{align*}
%
The following result is standard:
\begin{theorem}[\cite{CDLLR07,KR10our}]
For every \OWLQL{} KB $\mathcal{K}=(\mathcal{T}, \mathcal{A})$, every CQ $\q(\vec{x})$ and every $\vec{a} \subseteq \ind(\mathcal{A})$, $\mathcal{K} \models \q(\vec{a})$ iff $\CIK \models \q(\vec{a})$.
\end{theorem}


\section{Query Rewriting}

Let $\Sigma$ be a signature that can be used to formulate queries and ABoxes (remember that $\Sigma$ does not contain any built-in predicates).
%
%
Given an ABox $\mathcal{A}$ over $\Sigma$, define $\mathcal{I}_\mathcal{A}$ to be the interpretation whose domain consists of all individuals in $\Sigma$ (even if they are not in $\ind(\mathcal{A})$) and  such that $\mathcal{I}_\mathcal{A} \models E(\vec{a})$ iff $E(\vec{a}) \in \mathcal{A}$,  for all predicates $E(\vec{x})$.

Given a CQ $\q(\vec{x})$ and an \OWLQL{} TBox $\mathcal{T}$, a first-order formula $\q'(\vec{x})$ over $\Sigma$ is called an \emph{FO-rewriting for $\q(\vec{x})$ and $\mathcal{T}$ over $\Sigma$} if, for any ABox $\mathcal{A}$ over $\Sigma$ and any $\vec{a} \subseteq \ind(\mathcal{A})$, we have $(\mathcal{T},\mathcal{A}) \models \q(\vec{a})$ iff $\mathcal{I}_\mathcal{A} \models \q'(\vec{a})$.
If $\q'$ is an FO-rewriting of the form $\exists \vec{y}\, \varphi(\vec{x}, \vec{y})$, where $\varphi$ is built from atoms using only $\land$ and $\lor$, then we call $\q'(\vec{x})$ a \emph{positive existential rewriting for $\q(\vec{x})$ and $\mathcal{T}$ over $\Sigma$} (or a \emph{PE-rewriting}, for short). The \emph{size} $|\q'|$ of $\q'$ is the number of symbols in it.

All known FO-rewritings for CQs and \OWLQL{} ontologies are of exponential size in the worst case. More precisely, for any CQ $\q$ and any \OWLQL{} TBox $\mathcal{T}$, one can construct a PE-rewriting of size $O((|\mathcal{T}| \cdot |\q|)^{|\q|})$~\cite{CDLLR07,Perez-UrbinaMH09,Chortaras-etal2011,2011_Gottlob,KR10our}. One of the main results of this paper is that this lower bound cannot be substantially improved in general. On the other hand, we shall see that FO-rewritings can be superpolynomially more succinct than pure PE-rewritings.

%

We shall also consider query rewritings in the form of nonrecursive Datalog queries.
We remind the reader (for details see, e.g.,~\cite{CeriGT89}) that a \emph{Datalog program}, $\Pi$, is a finite set of Horn clauses
$$
\forall \vec{x}\, (A_1 \land \dots \land A_m \to A_0),
$$
where each $A_i$ is an atom of the form $P(t_1,\dots,t_l)$ and each $t_j$ is either a variable from $\vec{x}$ or a constant. $A_0$ is called the \emph{head} of the clause, and $A_1,\dots,A_m$ its \emph{body}. All variables occurring in the head $A_0$ must also occur in the body, i.e., in one of $A_1,\dots,A_m$.
A predicate $P$ \emph{depends} on a predicate $Q$ in $\Pi$ if $\Pi$ contains a clause whose head's predicate is $P$ and whose body contains an atom with predicate $Q$. A Datalog program $\Pi$ is called \emph{nonrecursive} if this dependence relation for $\Pi$ is acyclic. A \emph{nonrecursive Datalog query} consists of a nonrecursive Datalog program $\Pi$ and a \emph{goal} $G$, which is just a predicate.  Given an ABox $\mathcal{A}$, a tuple $\vec{a} \subseteq \ind (\mathcal{A})$ is called a \emph{certain answer} to $(\Pi,G)$ over $\mathcal{A}$ if $\Pi,\mathcal{A} \models G(\vec{a})$. The \emph{size} $|\Pi|$ of $\Pi$ is the number of symbols in $\Pi$.

We distinguish between \emph{pure} and \emph{impure} nonrecursive Datalog queries~\cite{BenediktG10}. In a \emph{pure query} $(\Pi,G)$, the clauses in $\Pi$ do not contain \emph{constant symbols} in their heads. One reason for considering only pure queries in the OBDA setting is that impure ones can have too much impact on the data. For example, an impure query can explicitly add a ground atom $A_0(\vec{a})$ to the database, which has nothing to do with the intensional knowledge in the background ontologies.
In fact, impure nonrecursive Datalog queries are known to be more succinct than pure ones.

Given a CQ $\q(\vec{x})$ and an \OWLQL{} TBox $\mathcal{T}$, a pure nonrecursive Datalog query $(\Pi,G)$ is called a \emph{nonrecursive Datalog rewriting for $\q(\vec{x})$ and $\mathcal{T}$ over $\Sigma$} (or an \emph{NDL-rewriting}, for short) if, for any ABox $\mathcal{A}$ over $\Sigma$ and any $\vec{a} \subseteq \ind(\mathcal{A})$, we have $(\mathcal{T},\mathcal{A}) \models \q(\vec{a})$ iff $\Pi,\mathcal{A} \models G(\vec{a})$ (note that $\Pi$ may define predicates that are not in $\Sigma$, but may not use non-signature constants). 
Similarly to FO-rewritings, known NDL-rewritings for \OWLQL{} are of exponential size~\cite{RosatiAKR10,2011_Gottlob}. Here we show that, in general, one cannot make NDL-rewritings shorter. On the other hand, NDL-rewritings can be exponentially more succinct than PE-rewritings.

The rewritings can be much shorter if non-signature predicates and constants become available. As follows from~\cite{GottlobS11}, every CQ over an \OWLQL{} ontology can be rewritten as a polynomial-size nonrecursive Datalog query if we can use the inequality predicate and  at least two distinct constants (cf.\ also \cite{Avigad01} which shows how two constants and $=$ can be used to eliminate definitions from first-order theories without an exponential blow-up). In fact, we observe that, using equality and two distinct constants, any CQ over an \OWLQL{} ontology can be rewritten into a PE-query of polynomial size.





\section{Boolean Functions and Circuits}\label{sec:circuits}

In this section we give a brief introduction to Boolean circuits and remind the reader of the results on monotone circuit complexity that we will use. 

An \emph{$n$-ary Boolean function}, for $n\ge 1$, is a function from $\{0,1\}^n$ to $\{0,1\}$. A Boolean function $f$ is \emph{monotone} if $f(\vec{\alpha}) \leq f(\vec{\alpha}')$, for all $\vec{\alpha}\leq \vec{\alpha}'$, where $\leq$ is the component-wise relation $\leq$ on vectors of $\{0,1\}$.


We remind the reader (for more details see, e.g.,~\cite{Arora&Barak09,Jukna12}) that an $n$-\emph{input Boolean circuit}, $\C$, is a directed acyclic graph with $n$ sources, \emph{inputs}, and one sink, \emph{output}. Every non-source node of $\C$ is called a \emph{gate}; it is labelled with either $\land$ or $\lor$, in which case it has two incoming edges, or with $\neg$, in which case it has one incoming edge. A circuit is \emph{monotone} if it contains only $\land$- and $\lor$-gates. We think of a \emph{Boolean formula} as a circuit in which every gate has at most one outgoing edge. For an input $\vec{\alpha} \in \{0,1\}^n$, the \emph{output} of $\C$ on $\vec{\alpha}$ is denoted by $\C(\vec{\alpha})$, and $\C$ is said to \emph{compute} an $n$-ary Boolean function $f$ if $\C(\vec{\alpha}) = f(\vec{\alpha})$, for every $\vec{\alpha} \in \{0,1\}^n$.
The number of nodes in $\C$ is the \emph{size} of $\C$, denoted by $|\C|$. 

%

A \emph{family of Boolean functions} is a sequence $f^1,f^2,\dots$,  where each $f^n$ is an $n$-ary Boolean function.
We say that a family $f^1,f^2,\dots$ is in the complexity class \NP{} if there exist polynomials $p$ and $T$ and, for each $n \geq 1$, a Boolean circuit $\C^n$ with $n + p(n)$ inputs such that  $|\C^n|\leq T(n)$ and, for each $\vec{\alpha} \in \{0,1\}^n$, we have
\begin{equation*}
f^n(\vec{\alpha}) =1\qquad\text{ iff }\qquad\C^n(\vec{\alpha},\vec{\beta})=1,\quad\text{for some }\vec{\beta} \in \{0,1\}^{p(n)}.
\end{equation*}
The additional $p(n)$ inputs for $\vec{\beta}$ in the $\C^n$ are called \emph{advice inputs}.

We shall use three well-known families of monotone Boolean functions in $\NP$:
\begin{description}
\item[\normalfont $\cli_{n,k}$] is the function of $n(n-1)/2$ variables $e_{ij}$, $1 \leq i < j\le n$, which returns 1 iff the graph with  vertices $\{1,\dots,n\}$ and  edges $\{ \{i,j\} \mid e_{ij}=1\}$ contains a $k$-clique. A series of papers,  started by Razborov's breakthrough~\cite{Razborov85}, gave an exponential lower bound for the size of monotone circuits computing $\cli(n,k)$: $2^{\Omega(\sqrt{k})}$ for all $k \leq \frac{1}{4} (n/ \log n)^{2/3}$~\cite{AlonB87}. For monotone formulas, an even better lower bound is known: $2^{\Omega(k)}$  for $k = 2n/3$~\cite{RazW92}. Since $\cli_{n,k}$ is $\NP$-complete, the question whether it can be computed by a polynomial-size Boolean circuit (i.e., belongs to the complexity class \Ppoly{}) is equivalent to whether $\NP \subseteq \Ppoly$, which is an open problem (see e.g.,~\cite{Arora&Barak09}).

It is not hard to see that $\cli_{n,k}$ can be computed by a nondeterministic circuit
of size $O(n^2)$ with $n$ advice variables: the circuit gets a vector $\vec{y} \in \{0,1\}^n$ indicating vertices of a clique as its advice inputs and checks whether the vector has $k$-many 1s and whether any two vertices given by 1s in advice inputs are indeed connected by an edge in the input graph. 

\item[\normalfont $\mat_{2n}$] is the function of $n^2$ variables $e_{ij}$, $1 \leq i,j \leq n$, which returns 1 iff there is a \emph{perfect matching} in the bipartite graph $G$ with vertices $\{v_1^1,\dots,v_n^1,v_1^2,\dots,v_n^2\}$ and edges $\{\{v_i^1,v_j^2\}\mid e_{ij} = 1\}$, i.e., a subset $E$ of edges in $G$ such that every node in $G$ occurs exactly once in $E$. An exponential lower bound $2^{\Omega(n)}$ for the size of monotone formulas computing $\mat_{2n}$ is known~\cite{RazW92}. On the other hand, this function is computable by non-monotone formulas of size $n^{O(\log n)}$~\cite{BorodinGH82}.

$\mat_{2n}$ can also be computed by a Boolean nondeterministic circuit of size $O(n^2)$ with $n^2$ advice variables: the circuit gets the edges a perfect matching in its advice inputs and it checks whether each edge in the perfect matching is an edge of the graph and whether, for each vertex, there is exactly one edge in the perfect matching containing it. 

\item[\normalfont $\gen_{n^3}$] is the function of $n^3$ variables $x_{ijk}$, $1 \le i,j,k \le n$, defined as follows. We say that $1$ \emph{generates} $k \le n$ if either $k=1$ or $x_{ijk}=1$ and $1$ generates both $i$ and  $j$. $\gen_{n^3}(x_{111},\dots, x_{nnn})$ returns $1$ iff $1$ generates $n$. $\gen_{n^3}$ is clearly a monotone Boolean function computable by polynomial-size monotone Boolean circuits. On the other hand, any monotone formula computing $\gen_{n^3}$ is of size  $2^{n^{\varepsilon}}$, for some $\varepsilon>0$~\cite{RazM97}.
\end{description}
The complexity results above will be used in Section~\ref{s:7} to obtain similar bounds for the size of rewritings for certain CQs and \OWLQL{} ontologies encoding these three function. The encoding will require a representation of these functions in terms of CNF.

\section{Circuits, CNFs and OBDA}
\label{sec:queries}

In this section we show how the above families of Boolean functions can be encoded as a CQ answering problem over \OWLQL{} ontologies.
More specifically, for each family $f^1,f^2,\dots$ of Boolean functions, we construct a sequence of \OWLQL{} TBoxes $\T_{f^n}$ and CQs $\q_{f^n}$, as well as ABoxes $\mathcal{A}_{\vec{\alpha}}$, $\vec{\alpha}\in\{0,1\}^n$, with a \emph{single} individual such that
\begin{equation*}
(\T_{f^n},\mathcal{A}_{\vec{\alpha}})\models \q_{f^n} \qquad\text{iff}\qquad f^n(\vec{\alpha})  =1,\qquad\text{ for all } \vec{\alpha}\in\{0,1\}^n.
\end{equation*}
Then we show that rewritings for $\q_{f^n}$ and  $\T_{f^n}$ correspond to Boolean circuits computing $f^n$. The construction proceeds in two steps: first, we represent the $f^n$ by polynomial-size CNFs (in a way similar to the Tseitin transformation~\cite{Tseitin83}),
and then encode those CNFs in terms of \OWLQL{} query answering.


Let $f^1,f^2,\dots$ be a family of Boolean functions in \NP{} and $\C^1,\C^2,\dots$ be a family  of circuits computing the $f^n$ (according to the definition above).
We consider the inputs $\vec{x}$ and the advice inputs $\vec{y}$ of $\C^n$ as Boolean variables; each of the gates $g_{1}, \dots, g_{\ell}$ of $\C^n$ is also thought of as a Boolean variable whose value coincides with the output of the gate on a given input. We assume that $\C^n$ contains only  $\neg$- and $\land$-gates, and so can be regarded as a set of equations of the form
\begin{equation*}
g_i = \neg h_i \qquad\text{or}\qquad g_i = h_i \land h_i',
\end{equation*}
where $h_i$ and $h_i'$ are the inputs of the gate $g_i$, that is, either input variables $\vec{x}$, advice variables $\vec{y}$ or other gates $\vec{g} = (g_1,\dots,g_{\ell})$. We assume $g_1$ to be the output of $\C^n$.
Now, with each $f^n$ and each $\vec{\alpha} = (\alpha_1,\dots,\alpha_n)\in\{0,1\}^n$, we associate the following formula in CNF:
\begin{multline*}
\phifn(\vec{x},\vec{y},\vec{g}) \ \ = \ \ \bigwedge_{\alpha_j = 0} \neg x_j \ \  \land \ \ g_1 \ \  \land \hspace*{-0.5em} \bigwedge_{g_i=\neg h_i \text{ in } \C^n} \hspace*{-0.5em}\bigl[(h_i \lor \neg g_i) \land (\neg h_i \lor g_i)\bigr] \ \land {} \\
\bigwedge_{g_i = h_i \land h_i' \text{ in } \C^n} \hspace*{-2em}\bigl[(h_i \vee \neg g_i)\land (h_i' \vee \neg g_{i})\land (\neg h_i \vee \neg h_i' \vee g_i)\bigr].
\end{multline*}
The clauses of the last two conjuncts encode the correct computation of the circuit: they are equivalent to $g_i \leftrightarrow \neg h_i$ and
$g_i \leftrightarrow h_i \land h_i'$, respectively.

\begin{lemma}\label{l:F}
If $f^n$ is a monotone Boolean function then  $f^n(\vec{\alpha}) = 1$ iff $\phifn$ is satisfiable,  for each $\vec{\alpha}\in\{0,1\}^n$.
\end{lemma}
\begin{proof}
$(\Rightarrow)$ Let $f^n(\vec{\alpha}) = 1$. Then $\C^n(\vec{\alpha}, \vec{\beta})=1$, for some $\vec{\beta}$. It can be easily seen that $\phifn(\vec{\alpha},\vec{\beta},\vec{\gamma}) = 1$, where the values $\gamma_i$ in $\vec{\gamma}$ are given by the outputs of the corresponding gates $g_i$ in $\C^n$ on the input $(\vec{\alpha}, \smash{\vec{\beta}})$.

$(\Leftarrow)$ Conversely, let
$\phifn(\vec{\alpha}',\vec{\beta},\vec{\gamma}) = 1$. By the first conjunct of $\phifn$,  $\vec{\alpha}' \leq \vec{\alpha}$.
As $f^n$ is monotone, it is enough to show $f^n(\vec{\alpha}')=1$. This is immediate from the second conjunct of $\phifn$, $g_1$, and an observation that the values $\vec{\gamma}$ are equal to the outputs of the corresponding gates of $\C^n$
on the input $(\vec{\alpha}', \smash{\vec{\beta}})$.
\end{proof}




The second step of the reduction is to encode satisfiability of $\phifn$ by means of the CQ answering problem in \OWLQL. Denote $\phifn$ for $\vec{\alpha} = (0,\dots,0)$ by $\phifn[]$. It is immediate from the definitions that, for each $\vec{\alpha}\in\{0,1\}^n$, the CNF $\phifn$ can be obtained from $\phifn[]$ by removing the clauses $\neg x_j$ for which $\alpha_j = 1$, $1 \leq j \leq n$.
The CNF $\phifn[]$ contains $d \leq 3|\C^n|$ clauses $C_1,\dots,C_d$ with $N = |\C^n|$ Boolean variables, which will be denoted by $p_1,\dots,p_N$.

Let $P$ be a role name and let $A_i$, $X_i^0$, $X_i^1$ and $Z_{i,j}$ be concept names. Consider the TBox $\mathcal{T}_{f^n}$ containing the following inclusions, for $1 \leq i \leq N$, $1 \leq j \leq d$:
\begin{align*}
& A_{i-1}  \sqsubseteq \exists P^-.X_i^\ell,\quad\text{ for } \ell = 0,1,\\
&X_i^\ell \sqsubseteq A_i, \qquad\text{ for } \ell = 0,1,\\
&\hspace*{2em} X_i^0 \sqsubseteq Z_{i,j}\ \ \text{ if } \ \ \neg p_i \in C_j,\\
&\hspace*{2em} X_i^1 \sqsubseteq Z_{i,j} \ \ \text { if } \ \ p_i \in C_j,\\
& Z_{i,j}  \sqsubseteq \exists P.Z_{i-1,j},\\
& A_0 \sqcap A_i \sqsubseteq \bot,\\
&
A_0 \sqcap \exists P \sqsubseteq \bot,\\
& A_0 \sqcap Z_{i,j} \sqsubseteq \bot, \ \ \text{ for } (i,j) \notin \{(0,1),\dots,(0,n)\}.\hspace*{-30mm}
\end{align*}
It can be seen that $|\T_{f^n}| = O(|\C^n|^2)$. Consider also the CQ
\begin{multline*}
 \q_{f^n} =  \exists \vec{y} \, \exists \vec{z} \ \Bigl[ A_0(y_0) \land  \bigwedge_{i = 1}^N P(y_i,y_{i-1}) \land{} \\
 \bigwedge_{j = 1}^d \Bigl(P(y_N,z_{N-1,j}) \land \bigwedge_{i = 1}^{N-1} P(z_{i,j},z_{i-1,j}) \land Z_{0,j}(z_{0,j}) \Bigr) \Bigr],
\end{multline*}
where $\vec{y} = (y_0,\dots,y_N)$ and
$\vec{z} = (z_{0,1},\dots,z_{N-1,1}, \dots, z_{0,d},\dots,z_{N-1,d})$. Clearly, $|\q_{f^n}| = O(|\C^n|^2)$. Note that $\mathcal{T}_{f^n}$ is acyclic and $\q_{f^n}$ is tree-shaped and has no answer variables.
For each $\vec{\alpha} = (\alpha_1,\dots,\alpha_n)\in\{0,1\}^n$, we set
\begin{equation*}
\mathcal{A}_{\vec{\alpha}} \ \ =  \ \ \bigl\{A_0(a) \bigr\} \ \ \cup \ \ \bigl\{ Z_{0,j}(a) \mid  1\leq j \leq n \text{ and } \alpha_j = 1 \bigr\}.
\end{equation*}

\begin{figure}[ht]
\begin{center}
\begin{tikzpicture}[>=latex, point/.style={circle,draw=black,minimum size=1.5mm,inner sep=0pt}]


\node at (-2.2,1) {$\mathcal{C}_{(\mathcal{T}_{f^n}, \mathcal{A}_{\vec{\alpha}})}$};
\node (a) at (-0.3,0) [point,label=below:{$a$},label=left:{\footnotesize $A_0,Z_{0,1}$}] {};
\tikzset{label distance=-1mm};
\node (a1) at (1.5,1.2) [point, label=below:{\footnotesize $\rule{0pt}{8pt}\hspace*{1.4em}X_1^1,\!Z_{1,3}$}] {};
\draw[<-,thick] (a) -- (a1);
\node (a0) at (1.5,-1.2) [point, label=below:{\footnotesize $X_1^0,Z_{1,1}$}] {};
\draw[<-,thick] (a) -- (a0);
\node (a11) at (3.3,1.8) [point, label=above:{\footnotesize $X_2^1$}] {};
\draw[<-,thick] (a1) -- (a11);
\node (a10) at (3.3,0.6) [point, label=above:{\footnotesize $X_2^0$}] {};
\draw[<-,thick] (a1) -- (a10);
\node (a01) at (3.3,-0.6) [point, label=above:{\footnotesize $X_2^1$}] {};
\draw[<-,thick] (a0) -- (a01);
\node (a00) at (3.3,-1.8) [point, label=above:{\footnotesize $X_2^0$}] {};
\draw[<-,thick] (a0) -- (a00);
\node (a111) at (5.1,2.1) [point, label=right:{\footnotesize $X_3^1$}] {};
\draw[<-,thick] (a11) -- (a111);
\node (a110) at (5.1,1.5) [point, label=right:{\footnotesize $X_3^0,Z_{3,3}$}] {};
\draw[<-,thick] (a11) -- (a110);
\node (a101) at (5.1,0.9) [point, label=right:{\footnotesize $X_3^1$}] {};
\draw[<-,thick] (a10) -- (a101);
\node (a100) at (5.1,0.3) [point, label=right:{\footnotesize $X_3^0,Z_{3,3}$}] {};
\draw[<-,thick] (a10) -- (a100);
\node (a011) at (5.1,-0.3) [point, label=right:{\footnotesize $X_3^1$}] {};
\draw[<-,thick] (a01) -- (a011);
\node (a010) at (5.1,-0.9) [point, label=right:{\footnotesize $X_3^0,Z_{3,3}$}] {};
\draw[<-,thick] (a01) -- (a010);
\node (a001) at (5.1,-1.5) [point, label=right:{\footnotesize $X_3^1$}] {};
\draw[<-,thick] (a00) -- (a001);
\node (a000) at (5.1,-2.1) [point, label=right:{\footnotesize $X_3^0,Z_{3,3}$}] {};
\draw[<-,thick] (a00) -- (a000);

\node (c11) at  (-0.3,-1.2) [point, label=left:{\footnotesize $Z_{0,1}$}] {}; %
\draw[<-] (c11) -- (a0);

\node (c31) at  (-0.3,1.2) [point, label=left:{\footnotesize $Z_{0,3}$}] {}; %
\draw[<-] (c31) -- (a1);


\foreach \y/\l/\a in {-2.1/00/below,-0.9/01/above,0.3/10/below,1.5/11/above}
{
\node (c000) at (3.3,\y) [point, label=below:{\footnotesize $Z_{2,3}$}] {}; %
\draw[<-] (c000) -- (a\l0);
\node (c000z) at (1.5,\y) [point,label=\a:{\footnotesize $Z_{1,3}$} ] {}; 
\draw[<-] (c000z) -- (c000);
\node (c000zz) at (-0.3,\y) [point, label=left:{\footnotesize $Z_{0,3}$}] {};
\draw[<-] (c000zz) -- (c000z);
}

\begin{scope}[yshift=-10mm]
\node at (-2.2,-3) {$\q_{f^n}$};
\node (y0) at (-0.3, -2.4) [point,label=above:{\footnotesize $y_0$},label=left:{\small $A_0$}] {};
\foreach \x/\n/\p in {1.5/1/0,3.3/2/1,5.1/3/2%
}
{
    \node (y\n) at (\x, -2.4) [point,label=above:{\footnotesize $y_\n$}] {};
    \draw[<-] (y\p) -- (y\n);
}
%
\node (z01) at (3.3,-3) [point,label=above:{\footnotesize $z_{2,j}$}] {};
\draw[->,out=-90,in=0] (y3) to  (z01);
\foreach \j/\y in {2/-3.3,3/-3.6,4/-3.9,5/-4.2}
{
	\node (z0\j) at (3.3,\y) [point] {}; 
	\draw[->,out=-90,in=0] (y3) to (z0\j);
}
\foreach \x/\n/\p/\m in {1.5/1/0/1, -0.3/2/1/0}
{
    \node (z\n1) at (\x, -3) [point,label=above:{\footnotesize $z_{\m,j}$}] {};
    \draw[->] (z\p1) -- (z\n1);
   \foreach \j/\y in {2/-3.3,3/-3.6,4/-3.9,5/-4.2}
   {
    	\node (z\n\j) at (\x, \y) [point] {}; 
    	\draw[->] (z\p\j) -- (z\n\j);
    }
}
\node (z 41) at (-0.3, -3) [label=left:{\small $Z_{0,1}$}] {};
\foreach \j/\y in {2/-3.3,3/-3.6,4/-3.9,5/-4.2}
{
	\node (z 4\j) at (-0.3, \y) [label=left:{\small $Z_{0,\j}$}] {};
}
\end{scope}
\end{tikzpicture}
\end{center}
\caption{Canonical model $\mathcal{C}_{(\mathcal{T}_{f^n}, \mathcal{A}_{\vec{\alpha}})}$ and query $\q_{f^n}$ for a Boolean function $f^n$, $n = 1$, computed by the circuit with one input $x$, one advice input $y$ and a single $\land$-gate. Thus,  $N = 3$, $d = 5$ and $\phifn[](x,y,g) = \neg x \land g \land (x \lor \neg g)\land (y \lor \neg g) \land (\neg x\lor\neg y\lor g)$. Points in $X_i^\ell$ are also in $A_i$, for all $1 \leq i \leq N$; the arrows denote role $P$ and the $Z_{i,j}$ branches in the canonical model are shown only for $j = 1,3$, i.e., for $\neg x$ and $(x\lor\neg g)$.}
\label{gen-mod}
\end{figure}
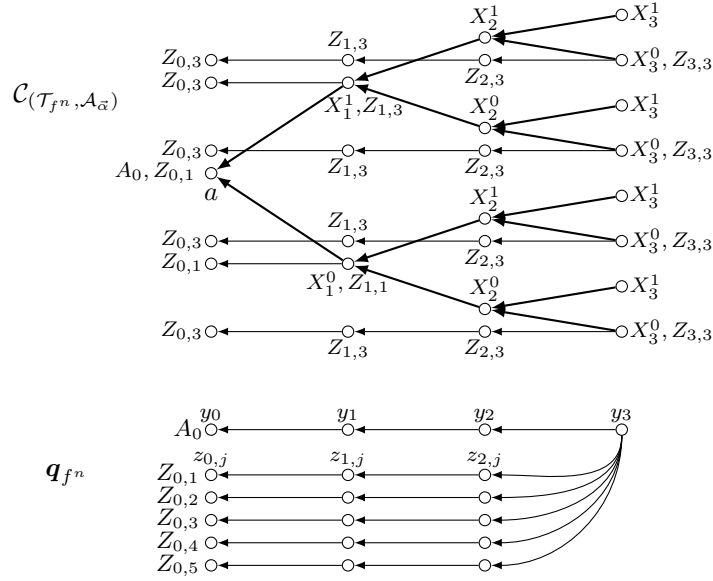

We explain the intuition behind the $\mathcal{T}_{f^n}$, $\q_{f^n}$ and $\mathcal{A}_{\vec{\alpha}}$ using the example of Fig.~\ref{gen-mod}, where the query $\q_{f^n}$  and the canonical model of $(\mathcal{T}_{f^n}, \mathcal{A}_{\vec{\alpha}})$, with $\mathcal{A}_{\vec{\alpha}} = \{A_0(a), Z_{0,1}(a) \}$,  are illustrated for some Boolean function. To answer $\q_{f^n}$ in the canonical model, we have to check whether $\q_{f^n}$ can be homomorphically mapped into it. The variables $y_i$ are clearly mapped to one of the branches of the canonical model from $a$ to a point in $A_3$, say the lowest one, which corresponds to the valuation for the variables in $\phifn$ making  all of them false. Now, there are two possible ways to map  variables $z_{2,1},z_{1,1},z_{0,1}$ that correspond to the clause $C_1 = \neg x_1$ in $\phifn[]$. If they are sent to the same branch so that $z_{0,1} \mapsto a$ then $Z_{0,1} (a)\in \mathcal{A}_{\vec{\alpha}}$, whence the clause $C_1$ cannot be in $\phifn$. Otherwise, they are mapped to the points in a side-branch so that $z_{0,1}\not\mapsto a$, in which case $\neg x_1$ must be true under our valuation. Thus, we arrive at the following:

\begin{lemma}\label{l1}
$(\mathcal{T}_{f^n},\mathcal{A}_{\vec{\alpha}}) \models \q_{f^n}$ iff $\phifn$ is satisfiable, for all $\vec{\alpha}\in\{0,1\}^n$.
\end{lemma}
\begin{proof}
%
$(\Rightarrow)$ Let $\mathfrak{a}$ be an assignment of points in the canonical model of $(\mathcal{T}_{f^n},\mathcal{A_{\vec{\alpha}}})$ to the variables of $\q_{f^n}$ under which it holds true. In particular, for all $0 \leq i \leq N$, $\mathfrak{a}(y_i)$ is in $A_i$, and thus the $\mathfrak{a}(y_i)$ define a vector $\gamma$ by taking $\gamma_i = 1$ if $\mathfrak{a}(y_i) \in X_i^1$ and $\gamma_i = 0$ otherwise, for all $1 \leq i \leq N$.
We show $\phifn(\vec{\gamma}) = 1$.
Take any clause $C_j$ in $\phifn[]$ and consider $\mathfrak{a} (z_{0,j}) \in Z_{0,j}$. If $\mathfrak a (z_{0,j}) = a$ then $j \leq n$, $Z_{0,j}(a) \in \mathcal{A}_{\vec{\alpha}}$ and $\alpha_j=1$; thus, the clause $x_j$ does not occur in $\phifn$. Otherwise, $\mathfrak a (z_{0,j}) \ne a$ and so, some $\mathfrak{a}(y_i)$ is in $Z_{i,j}$, which means that the clause $C_j$ contains $p_i$ if $\mathfrak{a}(y_i) \in X_i^1$ and $\neg p_i$ otherwise. By the definition of $\vec{\gamma}$, $\phifn(\vec{\gamma}) = 1$.

$(\Rightarrow)$ Suppose $\phifn(\vec{\gamma}) = 1$.
Recall that the canonical model of $(\mathcal{T}_{f^n},\mathcal{A}_{\vec{\alpha}})$ contains a 
path $u_0,\dots,u_N$ from $a = u_0$ to some $u_N$  that corresponds to that assignment in the following sense: for all $1 \leq i \leq N$, $u_i \in X_i^1$ if $\gamma_i = 1$ and $u_i \in X_i^0$ otherwise.
We construct an assignment $\mathfrak{a}$ of points in the canonical model of $(\mathcal{T}_{f^n},\mathcal{A}_{\vec{\alpha}})$ to the variables in $\q_{f^n}$ in accordance with this valuation. For $0 \leq i \le N$, we set $\mathfrak{a}(y_i) = u_i$. For $1 \leq j \leq m$, we define $\mathfrak{a}(z_{N-1,j}),\dots,\mathfrak{a}(z_{0,j})$ recursively, starting from $\mathfrak{a}(z_{N-1,j})$: set $\mathfrak{a}(z_{i,j}) = \mathfrak{a}(z_{i+1,j}) w_{[PZ_{i,j}]}$ if $\mathfrak{a}(z_{i+1,j})$ is in $Z_{i+1,j}$ and  $\mathfrak{a}(z_{i,j}) = u_i$, otherwise (assuming that $z_{N,j}=y_N$). It is easy to check that $\q_{f^n}$ is true in the canonical model under this assignment.
\end{proof}


\section{The Size of Rewritings}\label{s:6}

Now we show how PE-rewritings for $\q_{f^n}$ and $\T_{f^n}$ can be transformed into monotone Boolean formulas computing $f^n$, how FO-rewritings can be transformed into Boolean formulas and NDL-rewritings into monotone Boolean circuits.
%
%

\begin{lemma}\label{l:4}
Let $f^1,f^2,\dots$ be a family of monotone Boolean functions in \NP, and let $f = f^n$, for some $n$.

{\rm (i)} If $\q'_f$ is a PE-rewriting for $\q_f$ and $\T_f$  then there is a monotone Boolean formula $\psi_f$ computing $f$ with $|\psi_f| \le |\q'_f|$.

{\rm (ii)} If $\q_f'$ is an $FO$-rewriting for $\q_f$ and $\T_f$ and the signature $\Sigma$ contains a single constant then there is a Boolean formula $\psi_f$ computing $f$ with $|\psi_f| \le |\q'_f|$.

{\rm (iii)} If $(\Pi_f,G)$ is an NDL-rewriting for $\q_f$ and $\T_f$ then there is a monotone Boolean circuit $\C_f$ computing $f$ with $|\C_f| \le |\Pi_f|$.
\end{lemma}
\begin{proof}
{\rm (i)}
By Lemmas~\ref{l:F} and \ref{l1}, for any PE-rewriting  $\q'_{f}$ for $\q_{f}$ and $\mathcal{T}_{f}$, we have
\begin{equation*}
\mathcal{I}_{\mathcal{A}_{\vec{\alpha}}} \models \q'_{f} \qquad \text{iff} \qquad f(\vec{\alpha}) = 1, \quad \text{ for all }\vec{\alpha}\in\{0,1\}^n.
\end{equation*}
Recall that, of all ground atoms in signature $\Sigma$, only $A_0(a)$ and the $Z_{0,j}(a)$, for $1 \le j \le n$, can be true in $\mathcal{I}_{\mathcal{A}_{\vec{\alpha}}}$. In particular, no predicate can be true in $\mathcal{I}_{\mathcal{A}_{\vec{\alpha}}}$ on an element different from $a$. So, we can replace all the individual variables in $\q'_{f}$ with $a$, remove all (existential)  quantifiers and replace $A_0(a)$ by $\top$ and all the atoms different from $A_0(a)$ and $Z_{0,j}(a)$, for $1 \leq j\leq n$, by $\bot$ without affecting the truth-value of $\q'_f$ in $\mathcal{I}_\mathcal{A}$. Denote the resulting  PE-query by $\q_{f}^\dag$. It does not contain any variables and we have
$\mathcal{I}_{\mathcal{A}_{\vec{\alpha}}} \models \q_{f}^\dag$ iff $f(\vec{\alpha}) = 1$.
The formula $\q_{f}^\dag$ is equivalent to a propositional formula, $\psi_f$, with the connectives $\land$, $\lor$ and the propositional variables $Z_{0,j}(a)$, for $1 \leq j \leq n$, such that $\mathcal{I}_{\mathcal{A}_{\vec{\alpha}}} \models Z_{0,j}(a)$ iff $\alpha_j =  1$. Thus, $\psi_f$ computes $f$ and, clearly, $|\psi_f| \le |\q'_{f}|$.

{\rm (ii)}
If, in addition,  $\Sigma$ contains only one constant, $a$, then in the same way we can convert any FO-rewriting $\q'_{f}$ for $\q_{f}$ and $\mathcal{T}_{f}$ --- even with $\forall$ and $\neg$ --- to a propositional formula with variables $Z_{0,j}(a)$, for $1\le j \leq n$, which computes $f$.

{\rm (iii)}
Suppose now that $(\Pi_f, G)$ is an NDL-rewriting for $\q_{f}$ and $\mathcal{T}_{f}$ over a given signature $\Sigma$, containing $a$ among its constants. Then, for any ground $\Sigma$-atom $Q(t_1,\dots,t_l)$ with at least one $t_i$ different from $a$, we have $\Pi_f, \mathcal{A}_{\vec{\alpha}} \not\models Q(t_1,\dots,t_l)$ (which can be easily proved by induction of the length of derivations using the fact that $\Pi_f$ is pure and each variable that occurs in the head of a clause must also occur in its body). So we can again replace all the individual variables in $\Pi_f$ with $a$, $A_0(a)$ with $\top$ and all the atoms that do not occur in the head of a clause and different from $A_0(a)$ and $Z_{0,j}(a)$, for $1 \leq j \leq n$, with $\bot$.
Denote the resulting propositional NDL-program by $\Pi_f^\dag$. Then
$\Pi^\dag_f,\mathcal{A}_{\vec{\alpha}} \models G$ iff $f(\vec{\alpha}) = 1$.
The program $\Pi^\dag_f$ can now be transformed into a monotone Boolean circuit $\C_{f}$ computing $f$: for every (propositional) variable $p$ occurring in the head of a clause in $\Pi^\dag_f$, we introduce an $\lor$-gate whose output is $p$ and inputs are the bodies of the clauses with head $p$; and for each such body, we introduce an $\land$-gate whose inputs are the propositional variables in the body. The resulting  monotone Boolean circuit with inputs $Z_{0,j}(a)$, for $1 \le j \leq n$, and output $G$ is denoted by $\C_{f}$. Clearly, $|\C_{f}| \le |\Pi_f|$.
\end{proof}

\begin{lemma}\label{l:2}
Let $f^1,f^2,\dots$ be a family of monotone Boolean functions in \NP, and let $f = f^n$, for some $n$.
The following holds for signatures with a single constant\textup{:}

{\rm (i)} Suppose $\q'$ is an FO-sentence such that $(\T_f,\A_{\vec{\alpha}}) \models \q_f$ iff $\I_{\A_{\vec{\alpha}}} \models \q'$, for any $\vec{\alpha}$. Then\footnote{Here and below, $B(x)$ denotes $\exists y\,P(x,y)$ in the case of $B = \exists P$.}
$$
\q'' ~=~  \exists x\, \Bigl[A_0(x) \land \bigl(\q' \lor \bigvee_{A_0 \sqcap B \sqsubseteq_{\T_f} \bot}\hspace*{-1em} B(x)\bigr)\Bigr]
$$
is an FO-rewriting for $\q_f$ and $\T_f$ with $|\q''| = |\q'| + O(|\C^n|^2)$.

{\rm (ii)} Suppose $(\Pi,G)$ is a pure NDL-query with a propositional goal $G$ such that, $(\T_f,\A_{\vec{\alpha}}) \models \q_f$ iff $\Pi, \A_{\vec{\alpha}} \models G$, for any $\vec{\alpha}$. Then $(\Pi',G')$ is an NDL-rewriting for $\q_f$ and $\T_f$ with $|\Pi'| = |\Pi| + O(|\C^n|^2)$, where $G'$ is a fresh propositional variable and $\Pi'$ is obtained by extending $\Pi$ with the following  clauses\textup{:}
\begin{itemize}
\item[--] $\forall x \, (A_0(x) \land G \to G')$, 
\item[--] $\forall x\, (A_0(x) \land B(x) \to G')$, for all concepts $B$ such that $A_0 \sqcap B \sqsubseteq_{\T_f} \bot$.
\end{itemize}
\end{lemma}
\begin{proof}
(i) The queries $\q'$ and $\q''$ give the same answer over any $\A_{\vec{\alpha}}$. Consider a different ABox $\A'$ in the signature of $\T_f$ with $\ind(\mathcal{A}') = \{a\}$.
If $A_0(a) \notin \A'$ then we clearly have both $(\T_f,\A') \not\models \q_f$ and $\I_{\A'} \not\models \q''$.
If $\A'$ contains $A_0(a)$ and any ground atom in the signature of $\T_f$ different from $A_0(a), Z_{0,1}(a), \dots, Z_{0,n}(a)$ then $(\T_f,\A')$ is inconsistent, and so $(\T_f,\A') \models \q_f$. On the other hand, we clearly have $\I_{\A'} \models \q''$.

(ii) is proved in the same way.
The programs $(\Pi, G)$ and $(\Pi', G')$ give the same answer over any $\A_{\vec{\alpha}}$. Consider a different ABox $\A'$ in the signature of $\T_f$ with $\ind(\mathcal{A}') = \{a\}$.
If $A_0(a) \notin \A'$ then we clearly have both $(\T_f,\A') \not\models \q_f$ and
$\Pi', \A' \not\models G'$.
If $\A'$ contains $A_0(a)$ and any ground atom in the signature of $\T_f$ that is different from $A_0(a), Z_{0,1}(a), \dots, Z_{0,n}(a)$ then $(\T_f,\A')$ is inconsistent, and so $(\T_f,\A') \models \q_f$. On the other hand, we clearly have $\Pi', \A'   \models G'$.
\end{proof}

\begin{remark}\rm
It is worth noting that the lemma above can be extended to an \emph{arbitrary} signature (that is, to ABoxes with arbitrarily many individuals) provided that equality is available in rewritings. We refer to FO-rewritings with $=$ as FO$^=$-rewritings.

{\it
{\rm (i$'$)} Suppose that $\q'$ is an FO-sentence such that $(\T_f,\A_{\vec{\alpha}}) \models \q_f$ iff $\I_{\A_{\vec{\alpha}}} \models \q'$, for any $\vec{\alpha}$. Then there is an FO$^=$-rewriting $\q''$ for $\q_f$ and $\T_f$ such that $|\q''| \leq |\q'| + p(|\C^n|)$, for some polynomial $p$.

{\rm (ii$'$)} Suppose that $(\Pi,G)$ is a pure NDL-query with a propositional goal $G$ such that $(\T_f,\A_{\vec{\alpha}}) \models \q_f$ iff $\Pi, \A_{\vec{\alpha}} \models G$, for any $\vec{\alpha}$. Then there is an NDL-rewriting $(\Pi',G')$ for $\q_f$ and $\T_f$ such that $|\Pi'| \leq |\Pi| + p(|\C^n|)$, for some polynomial $p$.

}
\smallskip

\rm
The proof uses the polynomial `impure' PE- and NDL-rewritings of Section~\ref{s:equality} and~\cite{GottlobS11}.
To show (i$'$), let $\gamma$ be the PE-rewriting for $\q_f$ and $\T_f$ to be given in Section~\ref{s:equality}. We assume that this rewriting uses only two constants, say $0$ and $1$. Now, given an FO-sentence $\q'$ that is evaluated over ABoxes with a single individual only, we can clearly construct a quantifier-free FO-formula $\q_0(x)$ in the signature of $\q'$ such that it contains no constants and $\I_\A\models \q_0(a)$ iff $\I_\A\models \q'$, for all ABoxes with a single individual $a$.  Consider now the following FO-sentence
\begin{equation*}
\q'' ~=~  \exists x\, \Bigl[A_0(x) \land \bigl(\q_0(x) \ \ \lor \hspace*{-1em}\bigvee_{A_0 \sqcap B \sqsubseteq_{\T_f} \bot}\hspace*{-1em} B(x) \ \ \lor \ \ \exists y\,\bigr(P(y,x) \land \gamma[0/x, 1/y] \bigr)\Bigr],
\end{equation*}
where $\gamma[0/x, 1/y]$ is the result of replacing each occurrence of $0$ in $\gamma$ with $x$ and each occurrence of $1$ with $y$.

Suppose $(\T_f,\A)\models \q_f$. Then either $(\T_f,\A)$ is inconsistent or $\A$ has an individual $a_0$ such that $(\T_f,\A)\models \q_f(a_0)$, where $\q_f(a_0)$ is the query $\q_f$ with $y_0$ replaced by $a_0$.
In the former case, by the second disjunct, we have $\I_\A\models \q''$, which is a correct positive answer. In the latter case, if there is a distinct $a_1$ with $P(a_1,a_0)$ in $\A$ then the rewriting $\gamma$ provides the correct positive answer and, by the third disjunct, $\I_\A\models \q''$.
Finally, if neither of the above cases is applicable to $a_0$ then $\A_{a_0} = \{ D(a_0) \mid D(a_0)\in\A, D \text{ is a concept name} \}$ coincides with $\A_{\vec{\alpha}}$, for some $\vec{\alpha}$, in which case the correct positive answer is given by $\q_0(a_0)$.

Conversely, suppose $(\T_f,\A)\not\models \q_f$. Then $(\T_f,\A)$ is consistent, and so, the second disjunct is false. If there is no $a_0$ with $A_0(a_0)\in \A$ then, clearly, $\I_\A\not\models \q''$.  So, take an arbitrary individual  $a_0$ such that $A_0(a_0)\in \A$. If $P(a_1,a_0)\in \A$, for some $a_1$ (distinct from $a_0$ due to consistency) then, on the one hand, we have $\I_\A\not\models \gamma[0/a_0,1/a_1]$ and so, the third disjunct is false. On the other hand, if $\A_{a_0} = \{ D(a_0) \mid D(a_0)\in\A, D \text{ is a concept name} \}$ coincides with some $\A_{\vec{\alpha}}$ then $(\T_f,\A_{a_0})\models\q_f$ iff $\I_{\A_{a_0}}\models \q'$ iff $\I_{\A_{a_0}}\models \q_0(a_0)$. It follows that $\I_\A\not\models \q_0(a)$, for all individuals $a$ with $A_0(a)\in \A$, and so, the first disjunct is false as well.

\smallskip

Claim (ii$'$) is proved in a similar way, using a modification of the polynomial-size NDL-rewriting of Gottlob and Schwentick~\cite{GottlobS11}. (We note that in the short NDL-rewriting of~\cite{GottlobS11} the inequality predicate $\neq$ is applied only to terms that range over the extra constants, and not ABox individuals, and therefore one can write a short program defining $\neq$ by listing all pairs of non-equal constants.)  Let NDL-query $(\Delta,Q(z_0,z_1))$ be the short impure rewriting for $\q_f$ and $\T_f$, which uses $z_0$ and $z_1$ for the constants $0$ and $1$. Next, given an NDL-query $(\Pi,G)$ that is evaluated over ABoxes with a single individual only, we can construct a new NDL-query $(\Pi_0,G_0(x))$ such that all predicates of $\Pi_0$ are unary, all clauses have a single variable and $\Pi_0,\A\models G_0(a)$ iff $\Pi,\A\models G$, for all ABoxes with a single individual $a$.  Consider now $(\Pi',G')$, where $G'$ is a fresh propositional variable and $\Pi'$ consists of $\Pi_0$, $\Delta$ and the following three clauses:
\begin{itemize}
\item[--] $\forall x\,(A_0(x) \land B(x) \to G')$, for all concepts $B$ with $A_0 \sqcap B \sqsubseteq_{\T_f} \bot$,

\item[--] $\forall x\,(A_0(x) \land G_0(x) \to G')$,

\item[--] $\forall x,y\, (A_0(x)\land P(y,x)\land Q(x,y) \to G')$.
\end{itemize}
Suppose $(\T_f,\A)\models \q_f$. Then either $(\T_f,\A)$ is inconsistent or $\A$ has an individual $a_0$ such that $(\T_f,\A)\models \q_f(a_0)$, where $\q_f(a_0)$ is the query $\q_f$ with $y_0$ replaced by $a_0$.
In the former case, by the first clause, we have $\Pi',\A\models G'$, which is a correct positive answer. In the latter case, if there is a distinct $a_1$ with $P(a_1,a_0)$ in $\A$ then the program $\Delta$ provides the correct positive answer and, by the third clause, $\Pi',\A\models G'$.
Finally, if neither of the above cases is applicable to $a_0$ then $\A_{a_0} = \{ D(a_0) \mid D(a_0)\in\A, D \text{ is a concept name} \}$ coincides with some $\A_{\vec{\alpha}}$, in which case the correct positive answer is given by $\Pi_0$.

Conversely, suppose $(\T_f,\A)\not\models \q_f$. Then $(\T_f,\A)$ is consistent, and so, the first clause is not applicable. If there is no $a_0$ with $A(a_0)\in \A$ then, clearly, $\Pi',\A\not\models G'$.  So, take an arbitrary individual  $a_0$ such that $A(a_0)\in \A$. If $P(a_1,a_0)\in \A$, for some $a_1$ (distinct from $a_0$ due to consistency) then, on the one hand,  $\Delta,\A\not\models Q(a_0,a_1)$ and so, the third clause cannot give a positive answer. On the other hand, if $\A_{a_0} = \{ D(a_0) \mid D(a_0)\in\A, D \text{ is a concept name} \}$ coincides with some $\A_{\vec{\alpha}}$ then  $(\T_f,\A_{a_0})\models\q_f$ iff $\Pi,\A_{a_0}\models G$ iff $\Pi_0,\A_{a_0}\models G_0(a_0)$. It follows that $\Pi_0,\A\not\models G_0(a)$, for all individuals $a$ with $A_0(a)\in \A$, and so, the second clause cannot give a positive answer as well.
\end{remark}

We are in a position now to prove our main theorem which connects the size of circuits computing monotone Boolean functions with the size of rewritings for the corresponding queries and ontologies.

%
%
%

\begin{theorem}\label{thm.main}
For any family $f^1,f^2,\dots$ of monotone Boolean functions in \NP,  there exist polynomial-size CQs $\q_n$  and \OWLQL{} TBoxes $\T_n$ such that the following holds\textup{:}
\begin{description}
\item[\normalfont (1)] Let $L(n)$ be a lower bound for the size of monotone Boolean formulas computing $f^n$. Then,  $|\q'_n| \ge L(n)$, for any PE-rewriting $\q'_n$ for  $\q_n$ and $\T_n$.

\item[\normalfont (2)] Let $L(n)$ and $U(n)$ be a lower and an upper bound for the size of monotone Boolean circuits computing $f^n$. Then
\begin{itemize}
\item[--] $|\Pi_n| \ge L(n)$, for any NDL-rewriting $(\Pi_n,G)$ for $\q_n$ and $\T_n$\textup{;}

\item[--] there exist a polynomial $p$ and an NDL-rewriting $(\Pi_n,G)$ for $\q_n$ and $\T_n$  over any suitable signature with a single constant such that $|\Pi_n|\le U(n) + p(n)$.
\end{itemize}

\item[\normalfont (3)] Let $L(n)$ and $U(n)$ be a lower and an upper bound for the size of Boolean formulas computing $f^n$. Then
\begin{itemize}
\item[--] $|\q'_n| \ge L(n)$, for any FO-rewriting $\q'_n$ for $\q_n$  and $\T_n$ over any suitable signature with a single constant\textup{;}

\item[--] there exist a polynomial $p$ and an FO-rewriting $\q'_n$ for $\q_n$ and $\T_n$ over any suitable signature with a single constant with $|\q'_n| \le U(n) + p(n)$.
\end{itemize}
\end{description}
\end{theorem}
\begin{proof}
(1) follows from Lemma~\ref{l:4}~(i). The first claim of (2) from Lemma~\ref{l:4}~(ii). To prove the second claim, take any circuit $\C^n$ computing $f^n$ and having size $\le U(n)$. By Lemmas~\ref{l:F} and~\ref{l1},
$(\T_{f^n},\A_{\vec{\alpha}}) \models \q_{f^n}$ iff $\C^n(\vec{\alpha})=1$, for all $\vec{\alpha}\in\{0,1\}^n$. It should be clear that $\C^n$ can be transformed into an NDL-query $(\Pi,G)$ of size $|\C^n|$ such that $\Pi, \A_{\vec{\alpha}} \models G$ iff $(\T_{f^n},\A_{\vec{\alpha}}) \models \q_{f^n}$.
Then we apply Lemma~\ref{l:2}. (3)~is proved analogously.
\end{proof}


\section{Rewritings Long and Short}\label{s:7}

Now we apply Theorem~\ref{thm.main} to the Boolean functions mentioned in Section~\ref{sec:circuits} to demonstrate that some queries and ontologies may only have very long rewritings, and that rewritings of one type can be exponentially more succinct than rewritings of another type.

First we show that one cannot avoid an exponential blow-up for PE- and NDL-rewritings. We also show that even FO-rewritings can blow-up superpolynomially for signatures with a single constant under the assumption that $\NP \not\subseteq \Ppoly$.

\begin{theorem}\label{c1}
There is a sequence of CQs $\q_n$ of size $O(n)$ and $\OWLQL$ TBoxes $\T_n$ of size $O(n)$ such that\textup{:}
\begin{itemize}
\item[--] any PE-rewriting for $\q_n$ and $\T_n$ \textup{(}over any suitable signature\textup{)} is of size  $\geq 2^{\Omega(n^{1/4})}$\textup{;}

\item[--] any NDL-rewriting for $\q_n$ and $\T_n$ \textup{(}over any suitable signature\textup{)} is of size 
$\geq 2^{\Omega(({n/\log n})^{1/12})}$\textup{;}

\item[--] there does not exist a polynomial-size $FO$-rewriting for $\q_n$ and $\T_n$ over any suitable signature with a single constant unless $\NP \subseteq \Ppoly$.
\end{itemize}
\end{theorem}
\begin{proof}
Consider $f^n = \cli_{m,k}$ for $m = \lfloor n^{1/4} \rfloor$ and  $k = \lfloor 2m/3 \rfloor = \Omega(n^{1/4})$.
Then 
the size of $\q_n = \q_{f^n}$ and $\T_n = \T_{f^n}$ is $O(n)$.
The lower bound for PE-rewritings follows from Theorem~\ref{thm.main} and the lower bound for $\cli_{m,k}$~\cite{RazW92}.
The lower bound for NDL-rewritings is obtained by using a similar family with $k = \lfloor (m/ \log m)^{2/3} \rfloor  = \Omega((n/ \log n)^{1/6})$~\cite{AlonB87}.
If we assume $\NP \nsubseteq \Ppoly$ then there is no polynomial-size circuit for the \NP-complete function $\cli_{m,k}$, whence there is no polynomial-size FO-rewriting
of $\q_{f^n}$ and $\T_{f^n}$ over any signature containing a single constant.
\end{proof}

\begin{remark}\em
By the Karp-Lipton theorem (see, e.g.,~\cite{Arora&Barak09}) $\NP \subseteq \Ppoly$ implies $\PH = \Sigma_{2}^{p}$.
Thus, in Theorem~\ref{c1}, we can replace the assumption $\NP \not\subseteq \Ppoly$ with $\PH \neq \Sigma_{2}^{p}$.
\end{remark}

Next we show that NDL-rewritings can be exponentially more succinct than PE-rewritings.

\begin{theorem}
There is a sequence of CQs $\q_n$ of size $O(n)$ and $\OWLQL$
TBoxes $\T_n$ of size $O(n)$ for which there exists a polynomial-size NDL-rewriting
over a signature with a single constant, but any PE-rewriting over this signature is of size $\ge 2^{n^{\varepsilon}}$, for some $\varepsilon>0$.
\end{theorem}
\begin{proof}
Consider the family $\gen_{m^3}$. There is a polynomial $p$ and monotone Boolean circuits $\C^{m^3}$ computing $\gen_{m^3}$ with  $|\C^{m^3}|\leq p(m)$. It follows that, for each $n$, we can choose $m$ so that, for $f^n = \gen_{m^3}$, the size of both $\q_n = \q_{f^n}$ and $\T_n = \T_{f^n}$ is $O(n)$. In fact, $m = \Theta(n^{\delta})$,  for some $\delta >0$.
By Theorem~\ref{thm.main} and the lower bounds on the circuit complexity of $\gen_{m^3}$~\cite{RazM97},
there is a polynomial NDL-rewriting of $\q_n$ and $\mathcal{T}_n$, but
any PE-rewriting of $\q_n$ and $\mathcal{T}_n$ is of size $\geq 2^{n^{\varepsilon}}$, 
for some $\varepsilon > 0$.
\end{proof}

FO-rewritings can also be substantially shorter than the PE-rewritings:

\begin{theorem} \label{cor.matching1}
There is a sequence of CQs $\q_n$ of size $O(n)$ and $\OWLQL$ TBoxes $\T_n$ of size $O(n)$ which has an FO-rewriting of size $n^{O(\log n)}$ over a signature with a single constant, but any PE-rewriting over this signature is of size $\ge 2^{\Omega(n^{1/4})}$.
\end{theorem}
\begin{proof}
Consider $f^n = \mat_{2m}$ with $m = \lfloor n^{1/4} \rfloor$.
Then the size of both $\q_n = \q_{f^n}$ and $\T_n = \T_{f^n}$ is $O(n)$.
By Theorem~\ref{thm.main} and the bounds for circuit complexity of $\mat_{2m}$~\cite{RazW92,BorodinGH82}, we obtain the required
lower bound for PE-rewritings and the required upper bound for FO-rewritings; note that
$(n^{{1/4}})^{\log n^{1/4}} = n^{O(\log n)}$.
\end{proof}

In fact, we can use a standard trick from the circuit complexity theory to show that FO-rewritings can be superpolynomially more succinct than PE-rewritings.

\begin{theorem}
There is a sequence of CQs $\q_n$ of size $O(n)$ and $\OWLQL$ TBoxes $\T_n$ of size $O(n)$ which has a polynomial-size FO-rewriting over a signature with a single constant, but any PE-rewriting over this signature is of size $\ge 2^{\Omega(2^{\log^{1/2} n})}$.
\end{theorem}
\begin{proof}
Consider $f^n = \mat_{2m}$ with 
$m = \lfloor 2^{\log^{1/2} n}\rfloor$ variables and
add $\lfloor n^{1/4}\rfloor - m$ new dummy  variables to each $f^n$.
Then the size of both $\q_n = \q_{f^n}$ and $\T_n = \T_{f^n}$ is $O(n)$.
But now Theorem~\ref{thm.main} and the bounds for the circuit complexity of $\mat_{2m}$~\cite{RazW92,BorodinGH82} give the
$m^{O(\log m)} = n^{O(1)}$ upper bound for the size of FO-rewritings and the $2^{\Omega(m)} = 2^{
\Omega(2^{\log^{1/2}n})}$ lower bound for the size of PE-rewritings.
\end{proof}

\section{Short Impure Rewritings}\label{s:equality}

In the proof of Theorem~\ref{c1}, we used CQs 
containing no constant symbols. It follows that the theorem will still hold if we allow the built-in predicates $=$ and $\ne$ in the rewritings, but disallow the use of constants that \emph{do not occur in the original query}.
The situation changes drastically if $=$, $\ne$ and two additional constants, say 0 and 1, are allowed in the rewritings. As shown by Gottlob and Schwentick~\cite{GottlobS11}, in this case there is a polynomial-size NDL-rewriting for any CQ and \OWLQL{} TBox. Roughly, the rewriting uses the extra expressive resources to encode in a succinct way the part of the canonical model that is relevant to  answering the given query. We call rewritings of this kind   \emph{impure} (indicating thereby that they use predicates and constants that do not occur in the original query and ontology). In fact, using the ideas of~\cite{Avigad01} and~\cite{GottlobS11}, one can  construct an impure polynomial-size PE-rewriting for any CQ and \OWLQL{} TBox:
\begin{theorem}
For every CQ $\q$ and every \OWLQL{} TBox $\T$, there is an impure PE-rewriting $\q'$ for $\q$ and $\T$ whose size is polynomial in $|\q|$ and $|\T|$.
\end{theorem}
\begin{proof}
We illustrate the idea of the proof for a larger ontology language of tuple-generating dependencies (TGDs). CQ answering under TGDs is undecidable in general~\cite{BeeriVardi81}. However,  certain classes of TGDs (linear, sticky, etc.~\cite{CaliGL09,CaliGP10}) enjoy the  so-called polynomial witness property (PWP)~\cite{GottlobS11}, which guarantees that, for each CQ $\q$ and each set $\mathcal{T}$ of TGDs from the class, there is a number $N$ polynomial in $|\q|$ and $|\mathcal{T}|$ such that, for each ABox $\mathcal{A}$, there is a sequence of $N$ chase steps that entail $\q$. \OWLQL{} has PWP because its concept and role inclusions are special cases of linear TGDs.

So, suppose we have a set $\mathcal{T}$ of TGDs from a class enjoying PWP.
Without loss of generality we may assume that all predicates are of arity $L$, all TGDs have precisely $m$ atoms in the body and there is at most one existentially quantified variables in the head (see e.g.,~\cite{GottlobS11}), i.e., the TGDs are formulas of the form
\begin{equation*}
\forall \vec{x} \,\bigl(P_1(\vec{t}_1) \land \dots \land P_m(\vec{t}_m) \to \exists z\, P_0(\vec{t}_0)\bigr),
\end{equation*}
where each vector $\vec{t}_1,\dots,\vec{t}_m$ consists of $L$ (not necessarily distinct) variables from $\vec{x}$ (they are  universally quantified) and each of the $L$ variables of $\vec{t}_0$ either coincides with one of the $\vec{x}$ (in which case it is universally quantified) or equals $z$ (in which case it is existentially quantified).
Consider a CQ without free variables
\begin{equation*}
\q = \exists \vec{y}\,\bigwedge_{k = 1}^{|\q|} R_k(y_{k1},\dots,y_{kL}).
\end{equation*}
By PWP, there is a number $N$, polynomial in $|\q|$ and $|\mathcal{T}|$, such that, for any ABox $\mathcal{A}$, the query $\q$ is true on atoms of $N$ steps of the chase for $\mathcal{T}$ and $\mathcal{A}$ (provided that $(\mathcal{T},\mathcal{A})\models \q$).
In essence, our PE-rewriting guesses these $N$ ground atoms $\tau_1,\dots,\tau_N$ of the chase for $(\mathcal{T},\mathcal{A})$ and then checks whether the guess is a positive answer to $\q$ and the atoms indeed form steps of the chase for $(\mathcal{T},\mathcal{A})$. For each chase step $1 \leq i \leq N$, we will need the following variables:
\begin{itemize}
\item[--] $u_{i1},\dots,u_{iL}$ are the arguments of the ground atom $\tau_i$ and range over the ABox domain and the labelled nulls $\textit{null}_i$ (all these labelled nulls can be thought of as natural numbers not exceeding $N$);
\item[--] $r_i$ is the number of the predicate of $\tau_i$ (each predicate name $P$ is given a unique number, denoted by $[P]$); so, $r_i$ with $u_{i1},\dots,u_{iL}$ encode $\tau_i$;
\item[--] $w_{i1},\dots,w_{i\ell}$, where $\ell$ is the maximum length of the $\vec{x}$ in TGDs, are the arguments of the body of the TGD that generated $\tau_i$; they also range over the ABox domain and the labelled nulls (clearly, $\ell$ does not exceed $m\cdot L$).
\end{itemize}
The PE-rewriting is then defined by taking
\begin{equation*}
\q' = \exists \vec{y} \exists \vec{u}\exists \vec{r}  \exists \vec{w}\,\Bigl(\bigwedge_{k = 1}^{|\q|}\bigvee_{i = 1}^N \Bigl[(r_i = [R_k])\land \bigwedge_{j = 1}^L (u_{ij} = y_{kj})\Bigr] \land \bigwedge_{i = 1}^N \bigvee \Phi_i \Bigr).
\end{equation*}
The first conjunct of the rewriting chooses, for each atom in the query, one of the ground atoms $\tau_1,\dots,\tau_N$ in such a way that its predicate coincides with the query atom's predicate and the arguments match. The second conjunct chooses, for each ground atom $\tau_1,\dots,\tau_N$, the number of a TGD that produces it or 0, if the atom is taken from the ABox. So, the set of formulas $\Phi_i$ contains
\begin{equation*}
\bigvee_{P \text{ is a predicate}} \bigl((r_i = [P]) \land P(u_{i1},\dots,u_{iL})\bigr)
\end{equation*}
for the case when $\tau_i$ is taken from the ABox ($r_i$ is such that $P(u_{i1},\dots,u_{iL})$ is in the ABox for the predicate $P$ with the number $r_i$) and the following disjunct, for each TGD
\begin{equation*}
\forall \vec{x} \,\bigl(P_1(t_{11},\dots,t_{1L}) \land \dots \land P_m(t_{m1},\dots,t_{mL}) \to \exists z\, P_0(t_{01},\dots,t_{0L})\bigr)
\end{equation*}
in $\mathcal{T}$, modelling the corresponding chase rule application:
\begin{multline*}
(r_i = [P_0]) \land \bigwedge_{t_{0j} = x_l} (u_{ij} = w_{il}) \land\hspace*{-0em} \bigwedge_{t_{0j} = z}\hspace*{-0em} (u_{ij} = \textit{null}_i) \land {} \\
\bigwedge_{k = 1}^m \bigvee_{i' = 1}^{i-1} \bigl( (r_{i'} = [P_k]) \land \bigwedge_{t_{kj} = x_l} (w_{il} = u_{i'j})\bigr).
\end{multline*}
%
Informally, if $\tau_i$ is generated by an application of the TGD above, then $r_i$ is the number  $[P_0]$ of the head predicate $P_0$ and the existential variable $z$ of the head gets the unique null value $\textit{null}_i$ (third conjunct). Then, for each of the $m$ atoms of the body, one can choose a number $i'$ that is \emph{less than $i$} such that the predicate of $\tau_{i'}$ is the same as the predicate of the body atom and their arguments match (the last two conjuncts). The variables $w_{il}$ ensure that the same universally quantified variable gets the same value in different body atoms and in the head (if it occurs there, see the second conjunct).


It can be verified that $|\q'|=O(|\q|\cdot |\mathcal{T}|\cdot N^2 \cdot L)$ and that $(\mathcal{T},\mathcal{A})\models\q$ iff $\q'$ is true in the model $\mathcal{I}_\mathcal{A}$ extended with constants $1,\dots,N$ (these constants are distinct and do not belong to the interpretation of any predicate but $=$).

It should be noted that one can replace the numbers in the rewriting with just two  constants 0 and 1 (again,
with only $=$ interpreted over them). Each of the variables $u_{ij}$ can be replaced with a tuple $\bar{u}_{ij},u_{ij}^p,\dots,u_{ij}^0$ of variables with $p = \lceil\log N\rceil$ such that $\bar{u}_{ij}$ ranges over the ABox elements and $u_{ij}^p,\dots,u_{ij}^0$ range over $\{0,1\}$ and thus represent a number up to $N$. Similarly, we replace the $w_{il}$ and $r_i$. Each labelled null $\textit{null}_{i}$ is then replaced by the constant tuple representing the number $(i-1)$ in binary; the constants $[P]$ for the numbers of predicates $P$ are dealt with similarly. Finally, the equality atoms in the rewriting are replaced by the component-wise equalities and each $P(u_{i1},\dots,u_{iL})$ is replaced by $P(\bar{u}_{i1},\dots,\bar{u}_{iL})\land \bigwedge_{j = 1}^L\bigwedge_{k = 0}^p (u_{ij}^k = 0)$.
\end{proof}

Thus, we obtain the following:
\begin{theorem}
Impure PE- and NDL-rewritings for CQs and \OWLQL{} ontologies are exponentially more succinct than pure PE- and NDL-rewritings.
\end{theorem}



\section{Conclusion}

The exponential lower bounds for the size of `pure' rewritings above may look discouraging in the OBDA context. It is to be noted, however, that the ontologies and queries used their proofs are extremely `artificial' and never occur in practice (see the analysis in~\cite{KR12our}). As demonstrated by the existing description logic reasoners (such as FaCT\raisebox{1pt}{\scriptsize++}, HermiT, Pellet, Racer), \emph{real-world} ontologies can be classified efficiently despite the high worst-case complexity of the classification problem. We believe that practical query answering over \OWLQL{} ontologies can be feasible if supported by suitable optimisation and indexing techniques.  It also remains to be seen whether polynomial impure rewritings can be used in practice.

We conclude the paper by mentioning two open problems.
Our exponential lower bounds were proved for a sequence of pairs $(\q_n, \T_n)$. It is unclear whether these bounds hold uniformly for all $\q_n$ over the same $\T$:

\begin{question}
Do there exist an \OWLQL{} TBox $\T$ and CQs $\q_n$ such that any pure PE- or NDL-rewritings for $\q_n$ and $\T$ are of exponential size?
\end{question}
As we saw, both FO- and NDL-rewritings are more succinct than PE-rewritings.
\begin{question}
What is the relation between the size of FO- and NDL-rewritings?
\end{question}

%

\medskip
\noindent
{\bf Acknowledgments.} We thank the anonymous referees at ICALP 2012 for their constructive feedback and suggestions. This paper was supported by the U.K.~EPSRC grant EP/H05099X.

\end{document}